\documentclass[letterpaper, 10 pt, conference]{ieeeconf}  

\IEEEoverridecommandlockouts                              

\newcommand{\norm}[1]{\left\lVert#1\right\rVert}
\newcommand{\col}[1]{\text{col}\left(#1\right)}

\usepackage{graphics} 
\usepackage{epsfig} 
\usepackage{mathptmx} 
\usepackage{times} 
\usepackage{amsmath} 
\usepackage{amssymb}  
\usepackage{mathtools}
\usepackage{algorithm}
\usepackage{algorithmicx}
\usepackage{algpascal}
\usepackage{textcomp}

\newtheorem{lm}{Lemma}
\newtheorem{rmk}{Remark}
\newtheorem{thm}{Theorem}

\newtheorem{prop}{Proposition}

\title{\LARGE \bf
Data-Driven Prediction with Stochastic Data: Confidence Regions and Minimum Mean-Squared Error Estimates
}

\author{Mingzhou Yin, Andrea Iannelli, and Roy S. Smith
\thanks{This work was supported by the Swiss National Science Foundation under Grant 200021\_178890.}
\thanks{The authors are with the Automatic Control Laboratory, Swiss Federal Institute of Technology (ETH Zurich), Physikstrasse 3, 8092 Zurich, Switzerland,         {\tt\small \{myin,iannelli,rsmith\}@control.ee.ethz.ch.}}%
\thanks{This work has been submitted to the IEEE for possible publication. Copyright may be transferred without notice, after which this version may no longer be accessible.}%
}

\begin{document}

\maketitle
\thispagestyle{empty}
\pagestyle{empty}

\begin{abstract}

Recently, direct data-driven prediction has found important applications for controlling unknown systems, particularly in predictive control. Such an approach provides exact prediction using behavioral system theory when noise-free data are available. For stochastic data, although approximate predictors exist based on different statistical criteria, they fail to provide statistical guarantees of prediction accuracy. In this paper, confidence regions are provided for these stochastic predictors based on the prediction error distribution. Leveraging this, an optimal predictor which achieves minimum mean-squared prediction error is also proposed to enhance prediction accuracy. These results depend on some true model parameters, but they can also be replaced with an approximate data-driven formulation in practice. Numerical results show that the derived confidence region is valid and smaller prediction errors are observed for the proposed minimum mean-squared error estimate, even with the approximate data-driven formulation.

\end{abstract}
\section{Introduction}

In dynamical system analysis, one of the fundamental problems is to predict system responses from given inputs and initial conditions. Conventionally, this is done by simulating a model of the system, derived from first principles and/or experimental data. However, increasing complexity of systems poses challenges to the modeling process. Direct approaches have therefore been widely pursued to obtain reliable predictions of system responses without an explicit model \cite{IM-FD2}. In what follows, the term `data-driven' refers to such direct approaches.

A seminal result, known as the Willems' fundamental lemma \cite{Willems_2005}, shows that data-driven prediction can be conducted by linearly combining historical trajectory data with persistently exciting inputs for linear systems. A more general version of the lemma was recently given in \cite{IM-FD}. This result enables model-based control design techniques to be adopted with direct data-driven formulations. This framework is especially suitable for predictive control, where multiple data-driven algorithms have been developed, including subspace predictive control \cite{Sedghizadeh_2018}, data-enabled predictive control \cite{Coulson_2019}, and behavioral input-output parametrization \cite{furieri2021behavioral}. Successful applications have been described in \cite{lian2021adaptive,huang2021decentralized}.

Recently, the extension of the fundamental lemma to stochastic data from a system identification point of view has been drawing increasing interest \cite{dorfler2021bridging}. Such work includes model predictive control based on the prediction error method \cite{Huang_2019}, maximum likelihood signal matrix model \cite{yin2020maximum,pmlr-v144-yin21a}, and a Wasserstein distance minimization approach \cite{lian2021adaptive}.

With stochastic data, both the historical trajectories and the prediction conditions are uncertain, which makes it difficult to obtain statistical guarantees of the predictors. This limits the application of data-driven predictors to control design, particularly when robustness requirements and safety constraints exist. As a result, to the best of our knowledge, existing work on robust data-driven control \cite{alanwar2021robust,Berberich_2021,Berberich_2020} is restricted to bounded noise models with often loose prediction error bounds.

In this paper, a statistical framework on the accuracy of the predicted response is established under the assumption of Gaussian output noise. With this framework, confidence regions are available for a general form of stochastic data-driven predictors. The confidence region depends on the extended observability matrix of the system, but it can also be approximated through a data-driven formulation of model properties without direct knowledge of model parameters. The validity of the derived confidence regions is verified by numerical examples.

In addition, this statistical framework allows computation of the mean-squared error (MSE) of the predictor. In this way, a novel stochastic data-driven predictor is designed to be optimal for prediction accuracy in terms of minimizing the MSE. This optimal algorithm can be obtained in practice with a data-driven model characterization. It is shown numerically that the proposed minimum MSE predictor obtains smaller prediction errors than existing stochastic predictors.

\textit{Notation.} A Gaussian distribution with mean $\mu$ and covariance $\Sigma$ is indicated by $\mathcal{N}(\mu,\Sigma)$. The expectation and the covariance of a random vector $x$ are denoted by $\mathbb{E}(x)$ and $\text{cov}(x)$ respectively. For a vector $x$ and a positive definite matrix $Q$, the weighted Euclidean norm $(x^\mathsf{T}Qx)^{\frac{1}{2}}$ is denoted by $\norm{x}_Q$. For a matrix $X$, the vectorization operator $\text{vec}(X)$ stacks its columns in a single vector; $X^\dagger$ indicates the Moore-Penrose pseudoinverse. For a sequence of matrices $X_1,\dots,X_n$, we denote $[X_1^\mathsf{T}\ \dots\ X_n^\mathsf{T}]^\mathsf{T}$ by $\text{col}\left(X_1,\dots,X_n\right)$.

\section{The Data-Driven Prediction Problem}

\subsection{Problem Statement}
Consider a discrete-time linear time-invariant (LTI) system with output noise, given by
\begin{equation}
\begin{cases}
x_{t+1}&=\ A x_t+B u_t,\\
\hfil y_t&=\ C x_t + D u_t + w_t,
\label{eq:sys}
\end{cases}
\end{equation}
where $x_t \in \mathbb{R}^{n_x}$, $u_t \in \mathbb{R}^{n_u}$, $y_t \in \mathbb{R}^{n_y}$, $w_t \in \mathbb{R}^{n_y}$ are the states, inputs, outputs, and output noise respectively. In this paper, we assume that the system is observable with observability index (lag) $l$.

In data-driven prediction, the model parameters $A,B,C,D$ are unknown, but $M$ length-$L$ input-output trajectories
\begin{equation}
    z_i^d=\col{u_{t_i}^d,\cdots,u_{t_i+L-1}^d,y_{t_i}^d,\cdots,y_{t_i+L-1}^d}\in\mathbb{R}^{L(n_u+n_y)},
    \label{eqn:zi}
\end{equation}
where $i=0,\cdots,M-1$, have been collected. The matrix that concatenates these trajectories
\begin{equation}
    Z=\begin{bmatrix}z_0^d&\cdots&z_{M-1}^d\end{bmatrix}\in\mathbb{R}^{L(n_u+n_y)\times M}
\end{equation}
is termed the signal matrix \cite{yin2020maximum}. Depending on the construction, we can choose either $t_{i+1}=t_i+1$ for a mosaic Hankel signal matrix, or $t_{i+1}=t_i+L$ for a Page signal matrix \cite{damen1982approximate}. The trajectories can also come from independent experiments \cite{vanWaarde_2020}.

The problem is then to predict output trajectory $\mathbf{y}=\col{y_0,\cdots ,y_{L'-1}}$ from any given input trajectory $\mathbf{u}=\col{u_0,\cdots ,u_{L'-1}}$ using only the collected historical trajectories. To obtain a unique output trajectory, the initial condition is also fixed by measuring the immediate past input-output trajectory $\mathbf{u}_\text{ini}=\col{u_{-L_0},\cdots ,u_{-1}}$ and $\mathbf{y}_\text{ini}=\col{y_{-L_0},\cdots ,y_{-1}}$, where $L_0=L-L'\geq l$. In other words, the data-driven prediction problem aims to find an input-output mapping in the following form:
\begin{equation}
    \mathbf{y}=\mathcal{F}_{Z}(\mathbf{u};\mathbf{u}_\text{ini},\mathbf{y}_\text{ini}).
    \label{eqn:map}
\end{equation}

\subsection{Noise-Free Data-Driven Prediction}

In the noise-free case, the following lemma provides a condition for the existence of an exact mapping.
\begin{lm} 
    If $w_t=\mathbf{0}$, the exact mapping in the form of (\ref{eqn:map}) exists if $\text{rank}(Z)=n_u L+n_x$.
    \label{lm:1}
\end{lm}
\begin{proof}
According to Corollary 19 in \cite{IM-FD}, if $\text{rank}(Z)=n_u L+n_x$, for all $(\mathbf{u}_\text{ini},\mathbf{u},\mathbf{y}_\text{ini},\mathbf{y})$, there exists $g\in\mathbb{R}^M$, such that $\col{\mathbf{u}_\text{ini},\mathbf{u},\mathbf{y}_\text{ini},\mathbf{y}}=Zg$. Note that the observability index $l$ satisfies $n_y l\geq n_x$. The dimension of $\col{\mathbf{u}_\text{ini},\mathbf{u},\mathbf{y}_\text{ini}}$ then satisfies $n_u L+n_y L_0\geq\text{rank}(Z)$, so $\mathbf{y}$ can be uniquely determined by $(\mathbf{u}_\text{ini},\mathbf{u},\mathbf{y}_\text{ini})$.
\end{proof}

Define a partition of $Z$ as
\begin{equation}
    Z=\col{U_p,U_f,Y_p,Y_f},
\end{equation}
where $U_p\in\mathbb{R}^{n_u L_0\times M}$, $U_f\in\mathbb{R}^{n_u L'\times M}$, $Y_p\in\mathbb{R}^{n_y L_0\times M}$, $Y_f\in\mathbb{R}^{n_y L'\times M}$. Following the proof of Lemma~\ref{lm:1}, the mapping can be obtained by first solving the linear system
\begin{equation}
    \col{U_p,U_f,Y_p}g=\col{\mathbf{u}_\text{ini},\mathbf{u},\mathbf{y}_\text{ini}},
    \label{eqn:g1}
\end{equation}
and then applying $\mathbf{y}=Y_f g$. Although any solution to (\ref{eqn:g1}) is applicable (Proposition 1 in \cite{Markovsky_2008}), the pseudo-inverse solution is the most commonly used. So the solution to the noise-free data-driven prediction problem is,
\begin{equation}
    \mathcal{F}_{Z}(\cdot)=Y_f g_\text{pinv},\ g_\text{pinv}=\begin{bmatrix}U_p\\U_f\\Y_p\end{bmatrix}^\dagger\begin{bmatrix}\mathbf{u}_\text{ini}\\\mathbf{u}\\\mathbf{y}_\text{ini}\end{bmatrix}.
    \label{eqn:pinv}
\end{equation}

\subsection{Data-Driven Prediction with Stochastic Data}
\label{sec:sto}

When the output noise $w_t$ is no longer zero but a realization of a stochastic process, Lemma~\ref{lm:1} no longer holds and the mapping (\ref{eq:sys}) can only be estimated approximately. The output noise leads to uncertainties in both the output signal matrix $\col{Y_p,Y_f}$ and the output initial condition $\mathbf{y}_\text{ini}$. In this paper, the distribution of $w_t$ is assumed to be zero-mean Gaussian. Then, the distributions of $\mathbf{y}_\text{ini}$ and $\col{Y_p,Y_f}$ are also Gaussian. In what follows, the distributions are denoted by
\begin{equation}
\begin{split}
\mathbf{y}_\text{ini}&\sim \mathcal{N}\left(\mathbf{y}_\text{ini}^0,\Sigma_\text{yini}\right),\\
\text{vec}\left(\begin{bmatrix}Y_p\\Y_f\end{bmatrix}\right)&\sim \mathcal{N}\left(\text{vec}\left(\begin{bmatrix}Y_p^0\\Y_f^0\end{bmatrix}\right),\Sigma_Y\right),
\end{split}
\end{equation}
where $\mathbf{y}_\text{ini}^0$, $Y_p^0$, and $Y_f^0$ are noise-free versions of $\mathbf{y}_\text{ini}$, $Y_p$, and $Y_f$ respectively, and $\mathbf{y}_\text{ini}$ is uncorrelated with $\col{Y_p,Y_f}$.

Under this assumption, for a given $g$, the distribution of
\begin{equation}
    \begin{bmatrix}Y_p\\Y_f\end{bmatrix}g=\left(g^\mathsf{T}\otimes \mathbb{I}_{n_yL}\right)\text{vec}\left(\begin{bmatrix}Y_p\\Y_f\end{bmatrix}\right)
\end{equation}
is thus
\begin{equation}
    \left.\begin{bmatrix}Y_p\\Y_f\end{bmatrix}g\right| g\sim \mathcal{N}\Bigg(\begin{bmatrix}Y_p^0\\Y_f^0\end{bmatrix}g,\underbrace{\begin{bmatrix}\Sigma_{p}&\Sigma_{pf}\\\Sigma_{pf}^\mathsf{T}&\Sigma_{f}\end{bmatrix}}_{\Sigma_g}\Bigg),
\end{equation}
where $\Sigma_g=\left(g^\mathsf{T}\otimes \mathbb{I}_{n_yL}\right)\Sigma_Y\left(g\otimes \mathbb{I}_{n_yL}\right)$.

A special case of the noise model is when the noise is i.i.d. with $w_t\sim\mathcal{N}(\mathbf{0},\sigma^2\mathbb{I}_{n_y})$, and the signal matrix $Z$ is constructed as a Page matrix or from independent trajectories. In this case, we have $\Sigma_\text{yini}=\sigma^2\mathbb{I}_{n_yL_0}$, $\Sigma_Y=\sigma^2\mathbb{I}_{n_yLM}$, and thus $\Sigma_g=\sigma^2\norm{g}_2^2\,\mathbb{I}_{n_yL}$.

Different algorithms have been developed under this noise model, most of which share the following form:
\begin{subequations}
\begin{align}
    \mathcal{F}_{Z}(\cdot)&=Y_f g,\\
    \begin{bmatrix}U_p\\U_f\\Y_p\end{bmatrix}g&=\begin{bmatrix}\mathbf{u}_\text{ini}\\\mathbf{u}\\\mathbf{y}_\text{ini}+\delta\end{bmatrix}.
    \label{eqn:form2}%
\end{align}
\label{eqn:form}%
\end{subequations}
The slack variable $\delta$ is introduced to compensate for the error in both $Y_p$ and $\mathbf{y}_\text{ini}$. The algorithms then propose different strategies for balancing the magnitude of $g$ and the slack variable $\delta$. The algorithms are summarized as follows.

\textbf{Subspace predictor} \cite{Huang_2019,fiedler2021relationship}: the solution of the algorithm is exactly the same as that for the noise-free case (\ref{eqn:pinv}). However, the interpretation here is different. It corresponds to the least-squares estimate of a linear mapping:
\begin{equation}
    \mathcal{F}_{Z}(\cdot) = F_Z\, \col{\mathbf{u}_\text{ini},\mathbf{u},\mathbf{y}_\text{ini}},
    \label{eqn:sub}
\end{equation}
where
\begin{equation}
    F_Z =  \text{arg}\underset{F}{\text{min}}\ \norm{Y_f-F\,\col{U_p,U_f,Y_p}}_F^2.
\end{equation}
This coincides with finding the vector $g$ that minimizes $\norm{g}_2^2$ subject to (\ref{eqn:form2}) and $\delta=\mathbf{0}$.

\textbf{Signal matrix model} \cite{yin2020maximum,pmlr-v144-yin21a}: this algorithm uses maximum likelihood estimation to find the vector $g$ that maximizes the conditional probability of $\col{\delta,Y_f g}$ given $g$:
\begin{equation}
    \underset{g,\delta}{\text{min}}\ \text{logdet}\left(\Sigma_g+\begin{bmatrix}
    \Sigma_\text{yini}&\mathbf{0}\\\mathbf{0}&\mathbf{0}
    \end{bmatrix}\right)+\delta^\mathsf{T}\left(\Sigma_p+\Sigma_\text{yini}\right)^{-1}\delta,
    \label{eqn:smmexact}
\end{equation}
subject to (\ref{eqn:form2}). When $\Sigma_\text{yini}=\sigma^2\mathbb{I}_{n_yL_0}$ and $\Sigma_g=\sigma^2\norm{g}_2^2\,\mathbb{I}_{n_yL}$, an approximate quadratic program of (\ref{eqn:smmexact}) has been derived as
\begin{equation}
\begin{split}
    \underset{g,\delta}{\text{min}}&\ \norm{\delta}_2^2+n_y\left(L'\sigma^2/\norm{g_\text{pinv}}_2^2+L \sigma^2\right)\norm{g}_2^2\\
    \text{s.t.}&\qquad\qquad\qquad\quad\ \text{(\ref{eqn:form2})}.
\end{split}
    \label{eqn:smm}
\end{equation}

\textbf{Wasserstein distance minimization} \cite{lian2021adaptive}: this algorithm finds the vector $g$ that minimizes the Wasserstein distance between the stochastic distribution of $\mathbf{y}_\text{ini}$ and that of $Y_p g$:
\begin{equation}
    \underset{g,\delta}{\text{min}}\ \norm{\delta}_2^2+\text{tr}\left(\Sigma_\text{yini}+\Sigma_{p}-2\left(\Sigma_\text{yini}\Sigma_p\right)^{1/2}\right),
    \label{eqn:wexact}
\end{equation}
subject to (\ref{eqn:form2}). When $\Sigma_\text{yini}=\sigma^2\mathbb{I}_{n_yL_0}$ and $\Sigma_g=\sigma^2\norm{g}_2^2\,\mathbb{I}_{n_yL}$, an approximate quadratic program of (\ref{eqn:wexact}) has been derived as
\begin{equation}
\begin{split}
    \underset{g,\delta}{\text{min}}&\ 
    \norm{\delta}_2^2+n_y L_0\sigma^2\norm{g}_2^2\\
    \text{s.t.}&\qquad\quad\ \text{(\ref{eqn:form2})}.
\end{split}
    \label{eqn:w}
\end{equation}

It is noted that the algorithms (\ref{eqn:sub}), (\ref{eqn:smm}), and (\ref{eqn:w}) can be expressed in the following unified form:
\begin{equation}
    \begin{split}
    \mathcal{F}_{Z}(\cdot)=Y_f\ \text{arg}\underset{g}{\text{min}}&\  \norm{\delta}_2^2+\lambda\norm{g}_2^2\\
    \text{s.t.}&\qquad\text{(\ref{eqn:form2})},
    \end{split}
    \label{eqn:uni}
\end{equation}
where $\lambda\rightarrow 0$ for (\ref{eqn:sub}), $\lambda=n_y\left(L'\sigma^2/\norm{g_\text{pinv}}_2^2+L \sigma^2\right)$ for (\ref{eqn:smm}), and $\lambda=n_y L_0\sigma^2$ for (\ref{eqn:w}). With an abuse of notation, $\text{argmin}_g$ denotes the optimal solution of $g$ for the program depending on both $g$ and $\delta$. The optimization problem in (\ref{eqn:uni}) is a strongly convex quadratic program with only equality constraints. It admits a closed-form solution that is linear with respect to $\col{\mathbf{u}_\text{ini},\mathbf{u},\mathbf{y}_\text{ini}}$.

\section{Confidence Region Analysis}

In this section, confidence regions are established for the stochastic data-driven prediction algorithms discussed in Section~\ref{sec:sto}. The result first exploits information from the underlying state-space model. Then, a data-driven approximation of the model information is proposed.

\subsection{Derivation of the Confidence Region}

For any stochastic data-driven predictor in the form of (\ref{eqn:form}), the output estimate (\ref{eqn:map}) differs from the true output $\mathbf{y}_0$ due to the following two sources of error: 1) the output part of the signal matrix $Y_f$ is noisy, 2) the predictor estimates a trajectory whose output initial condition is $Y_p^0 g$, which differs from the trajectory to be predicted whose output initial condition is $\mathbf{y}_\text{ini}^0$. By characterizing the distributions of these two sources of error for a particular estimate of $g$ and $\delta$, we obtain the following confidence region for stochastic data-driven prediction.
\begin{thm}
Consider a stochastic data-driven predictor $\mathbf{y}=\mathcal{F}_{Z}(\mathbf{u};\mathbf{u}_\text{ini},\mathbf{y}_\text{ini})=Y_f g$ satisfying (\ref{eqn:form}). The true output $\mathbf{y}_0$ is in the following ellipsoidal set w.p. $p$:
\begin{equation}
\mathcal{Y}=\left\{\tilde{y}\mid\left(\mathbf{y}-\tilde{y}-\Gamma\delta\right)^\mathsf{T}\Sigma^{-1}\left(\mathbf{y}-\tilde{y}-\Gamma\delta\right)\leq \mu_p\right\},
\label{eqn:bound}
\end{equation}
where
\begin{equation}
    \Gamma=\col{CA^{L_0},\cdots,CA^{L-1}}\,\col{C,\cdots,CA^{L_0-1}}^\dagger,
    \label{eqn:gam}
\end{equation}
\begin{equation}
    \Sigma = \begin{bmatrix}
    -\Gamma&\mathbb{I}_{n_yL'}
    \end{bmatrix}
    \Sigma_g
    \begin{bmatrix}
    -\Gamma^\mathsf{T}\\\mathbb{I}_{n_yL'}
    \end{bmatrix}+
    \Gamma\,\Sigma_\text{yini}\Gamma^\mathsf{T},
    \label{eqn:sig}
\end{equation}
and $\mu_p$ satisfies $F_{\chi^2(L')}(\mu_p)=p$, where $F_{\chi^2(d)}(\cdot)$ is the cumulative distribution function of the $\chi^2$-distribution with $d$ degrees of freedom.
\label{thm:1}
\end{thm}
\begin{proof}
Let the stochastic noise in $Y_p$, $Y_f$, and $\mathbf{y}_\text{ini}$ be $E_p$, $E_f$, and $\epsilon_\text{ini}$ respectively, i.e.,
\begin{equation}
    E_p=Y_p-Y_p^0,\ E_f=Y_f-Y_f^0,\  \epsilon_\text{ini}=\mathbf{y}_\text{ini}-\mathbf{y}_\text{ini}^0.
    \label{eqn:error}
\end{equation}
The estimation error can be decomposed as follows, according to the two aforementioned sources of error
\begin{equation}
    \mathbf{y}-\mathbf{y}_0=E_fg + \mathbf{y}^-,
    \label{eqn:23}
\end{equation}
where $\mathbf{y}^-$ is the error due to the discrepancy $\left(Y_p^0 g-\mathbf{y}_\text{ini}^0\right)$ in the output initial condition. The initial condition error $\mathbf{y}^-$ can be seen as the free response from initial condition $\mathbf{u}^-_\text{ini}=\mathbf{0}$, $\mathbf{y}^-_\text{ini}=Y_p^0 g-\mathbf{y}_\text{ini}^0$. From (\ref{eqn:form2}) and (\ref{eqn:error}), we have
\begin{equation}
    Y_p^0 g = \mathbf{y}_\text{ini} + \delta - E_pg,\ \mathbf{y}_\text{ini}^0=\mathbf{y}_\text{ini} - \epsilon_\text{ini},
\end{equation}
\begin{equation}
    \mathbf{y}^-_\text{ini} = (\mathbf{y}_\text{ini} + \delta - E_pg) - (\mathbf{y}_\text{ini} - \epsilon_\text{ini}) = \delta + \epsilon_\text{ini} - E_pg.
    \label{eqn:25}
\end{equation}
Let the state of the trajectory at time $-L_0$ be $x^-$. Then we have
\begin{equation}
    \mathbf{y}^-_\text{ini} = \begin{bmatrix}C\\\vdots\\CA^{L_0-1}\\\end{bmatrix}x^-,\ 
    \mathbf{y}^- = \begin{bmatrix}CA^{L_0}\\\vdots\\CA^{L-1}\\\end{bmatrix}x^-.
\end{equation}
Since $L_0\geq l$, $\col{C,\cdots,CA^{L_0-1}}$ has full column rank. Thus, we have $x^-=\col{C,\cdots,CA^{L_0-1}}^\dagger \mathbf{y}^-_\text{ini}$. This directly leads to $\mathbf{y}^-=\Gamma\, \mathbf{y}^-_\text{ini}$. From (\ref{eqn:23})-(\ref{eqn:25}), the estimation error is then
\begin{equation}
    \mathbf{y}-\mathbf{y}_0=E_fg + \Gamma\left(\delta + \epsilon_\text{ini} - E_pg\right).
\end{equation}
Recall that $\epsilon_\text{ini}\sim\mathcal{N}\left(\mathbf{0},\Sigma_\text{yini}\right)$, $\left.\col{E_p,E_f}g\right| g\sim \mathcal{N}\left(\mathbf{0},\Sigma_g\right)$, and they are uncorrelated. The distribution of $(\mathbf{y}-\mathbf{y}_0)$ given $g$ and $\delta$ is Gaussian with
\begin{equation}
\begin{split}
    \mathbb{E}\left(\mathbf{y}-\mathbf{y}_0\right)&=\Gamma\delta,\\
    \text{cov}\left(\mathbf{y}-\mathbf{y}_0\right)&=\mathbb{E}\left(\begin{bmatrix}
    -\Gamma&\mathbb{I}_{n_yL'}
    \end{bmatrix}\begin{bmatrix}
    E_p\\E_f
    \end{bmatrix}g+\Gamma\epsilon_\text{ini}\right)\\&\qquad\qquad\left(\begin{bmatrix}
    -\Gamma&\mathbb{I}_{n_yL'}
    \end{bmatrix}\begin{bmatrix}
    E_p\\E_f
    \end{bmatrix}g+\Gamma\epsilon_\text{ini}\right)^\mathsf{T}\\
    &=\begin{bmatrix}
    -\Gamma&\mathbb{I}_{n_yL'}
    \end{bmatrix}
    \Sigma_g
    \begin{bmatrix}
    -\Gamma^\mathsf{T}\\\mathbb{I}_{n_yL'}
    \end{bmatrix}+
    \Gamma\,\Sigma_\text{yini}\Gamma^\mathsf{T}=\Sigma.
\end{split}
\label{eqn:dist}
\end{equation}
Therefore, $\left(\mathbf{y}-\mathbf{y}_0-\Gamma\delta\right)^\mathsf{T}\Sigma^{-1}\left(\mathbf{y}-\mathbf{y}_0-\Gamma\delta\right)$ is subject to the $\chi^2$-distribution with $L'$ degrees of freedom. This directly leads to (\ref{eqn:bound}).
\end{proof}

\begin{rmk}
Theorem~\ref{thm:1} stills holds when the system is not observable by replacing $A$, $C$, and $l$ with those for the observable part of the system.
\end{rmk}
\begin{rmk}
The derivation is inspired by the prediction error bound presented in Section~IV.C of \cite{Berberich_2021}. However, the results of \cite{Berberich_2021} consider a bounded non-stochastic noise model and provide a deterministic but admittedly non-tight bound on $\norm{\mathbf{y}-\mathbf{y}_0}$.
\end{rmk}

\subsection{Data-Driven Formulation of System Parameter $\Gamma$}
The confidence region given in Theorem~\ref{thm:1} is not available in practice since $\Gamma$ is dependent on the unknown model parameters $A$ and $C$. However, this system parameter matrix can be alternatively formulated by another data-driven prediction scheme offline.

As can be seen from the proof of Theorem~\ref{thm:1}, the matrix $\Gamma$ can be considered as a linear data-driven predictor with $\mathbf{u}=\mathbf{0}$ and $\mathbf{u}_\text{ini}=\mathbf{0}$. Supposing we have a noise-free signal matrix, the following lemma gives a data-driven version of Theorem~\ref{thm:1} without knowledge of $A$ and $C$.
\begin{lm}
    Let $\bar{Z}=\col{\bar{U}_p,\bar{U}_f,\bar{Y}_p,\bar{Y}_f}$ be a noise-free signal matrix with $\text{rank}(\bar{Z})=n_u L+n_x$. If $\Gamma$ is replaced by $\Gamma_Z=\bar{Y}_f P$, where $P$ is the last $n_y L_0$ columns of $\col{\bar{U}_p,\bar{U}_f,\bar{Y}_p}^\dagger$, then Theorem~\ref{thm:1} holds.
\end{lm}
\begin{proof}
    According to Lemma~\ref{lm:1}, for any output initial condition $\mathbf{y}_\text{ini}$, $\mathbf{y}=\Gamma_Z\,\mathbf{y}_\text{ini}$ is the unique free response with $\mathbf{u}_\text{ini}=\mathbf{0}$. So we have $\mathbf{y}^-=\Gamma_Z\, \mathbf{y}^-_\text{ini}$. The rest of the proof of Theorem~\ref{thm:1} remains the same.
\end{proof}
\begin{rmk}
    In general, $\Gamma_Z\neq\Gamma$. This is because when $n_y L_0>n_x$, the valid $\Gamma$ in the proof of Theorem~\ref{thm:1} is not unique. The pseudo-inverse solution (\ref{eqn:gam}) gives only one possibility.
\end{rmk}

In practice, the noisy signal matrix $Z$ can be used to find an approximation of the data-driven system parameter $\Gamma_Z$. Recall that the estimated mappings in the form of (\ref{eqn:uni}) admit linear solutions. So they can be employed to find an estimate of the linear mapping $\Gamma_Z$ by setting $\mathbf{u}=\mathbf{0}$, $\mathbf{u}_\text{ini}=\mathbf{0}$. The closed-form solution is given by
\begin{equation}
    \hat{\Gamma}_Z = Y_f\left(F^{-1}-F^{-1}U^\mathsf{T}(U F^{-1}U^\mathsf{T})^{-1}UF^{-1}\right)Y_p^\mathsf{T},
    \label{eqn:est}
\end{equation}
where $F=\lambda\mathbb{I}_M+Y_p^\mathsf{T}Y_p$ and $U=\col{U_p,U_f}$ as derived in \cite{yin2020maximum}. The hyperparameter $\lambda$ can be selected as approaching 0 (subspace predictor), $n_y L\sigma^2$ (signal matrix model), or $n_y L_0\sigma^2$ (Wasserstein distance minimization), and the corresponding $\hat{\Gamma}_Z$ estimates are denoted by $\hat{\Gamma}_\text{Sub}$, $\hat{\Gamma}_\text{SMM}$, and $\hat{\Gamma}_\text{WD}$ respectively. Note that in this estimation, the output initial condition $\mathbf{y}^-_\text{ini}$ is known exactly without noise. This leads to a slight change in the hyperparameter of the signal matrix model solution. When $\Gamma$ is replaced by $\hat{\Gamma}_Z$, Theorem~\ref{thm:1} only holds approximately. The validity of the approximation will be investigated in Section~\ref{sec:num}.

\section{Minimum Mean-Squared Error Algorithm}

In the proof of Theorem~\ref{thm:1}, the distribution of the estimation error has been derived in order to quantify the confidence region for a given estimate of $g$ and $\delta$ with the algorithms discussed in Section~\ref{sec:sto}. In this section, this distribution is used to propose a novel optimal predictor in the form of (\ref{eqn:form}), which directly targets maximum prediction accuracy, instead of the statistical properties as in Section~\ref{sec:sto}. This algorithm finds $g$ and $\delta$ in the mapping by minimizing the expected estimation error subject to (\ref{eqn:dist}), which leads to the following proposition.
\begin{prop}
    The minimum MSE estimate of the mapping in the form of (\ref{eqn:form}) is given by
    \begin{equation}
        \begin{split}
        \mathcal{F}_{Z}(\cdot)=Y_f\ \text{arg}\underset{g}{\text{min}}&\     \delta^\mathsf{T}\Gamma^\mathsf{T}\Gamma\delta+\text{tr}\left(\begin{bmatrix}
        -\Gamma&\mathbb{I}_{n_yL'}
        \end{bmatrix}
        \Sigma_g
        \begin{bmatrix}
        -\Gamma^\mathsf{T}\\\mathbb{I}_{n_yL'}
        \end{bmatrix}\right)\\
        \text{s.t.}&\qquad\qquad\qquad\quad\text{(\ref{eqn:form2})}.
        \end{split}
        \label{eqn:minmse}
    \end{equation}
    \label{prop:1}
\end{prop}
\begin{proof}
    From (\ref{eqn:dist}), we have
    \begin{equation}
    \begin{split}
        \text{MSE}\left(\mathbf{y}-\mathbf{y}_0\right)&=\mathbb{E}\left(\mathbf{y}-\mathbf{y}_0\right)^\mathsf{T}\left(\mathbf{y}-\mathbf{y}_0\right)\\
        &=\text{tr}\left(\text{cov}\left(\mathbf{y}-\mathbf{y}_0\right)+\mathbb{E}\left(\mathbf{y}-\mathbf{y}_0\right)\mathbb{E}\left(\mathbf{y}-\mathbf{y}_0\right)^\mathsf{T}\right)\\
        &=\text{tr}\left(\Sigma+\Gamma\delta\delta^\mathsf{T}\Gamma^\mathsf{T}\right)=\text{tr}\left(\Sigma\right)+\delta^\mathsf{T}\Gamma^\mathsf{T}\Gamma\delta.
    \end{split}
    \label{eqn:mse}
    \end{equation}
    where the third equality comes from (\ref{eqn:dist}). From the definition of $\Sigma$ in (\ref{eqn:sig}), it is observed that since $\Gamma\,\Sigma_\text{yini}\Gamma^\mathsf{T}$ does not depend on the optimization variables $g$ and $\delta$, minimizing the MSE is equivalent to the optimization problem in (\ref{eqn:minmse}).
\end{proof}

If we assume that $\Sigma_Y=\sigma^2\mathbb{I}_{n_yLM}$, (\ref{eqn:minmse}) becomes
\begin{equation}
    \begin{split}
    \mathcal{F}_{Z}(\cdot)=Y_f\ \text{arg}\underset{g}{\text{min}}&\    \norm{\delta}_Q^2+\lambda_\text{MSE}\norm{g}_2^2\\
    \text{s.t.}&\qquad\quad\text{(\ref{eqn:form2})},
    \label{eqn:mseapp}
    \end{split}
\end{equation}
where $Q=\Gamma^\mathsf{T}\Gamma$ and $\lambda_\text{MSE}=\sigma^2 n_yL'+\sigma^2\,\text{tr}\left(Q\right)$. This optimization problem is very similar to the unified form (\ref{eqn:uni}) for existing algorithms, except that the Euclidean norm of $\delta$ is now weighted by $Q$. The solution (\ref{eqn:mseapp}) is also linear with respect to $\col{\mathbf{u}_\text{ini},\mathbf{u},\mathbf{y}_\text{ini}}$.

The implications of Proposition~\ref{prop:1} are twofold. On the one hand, it provides the optimal solution to the data-driven prediction problem with output noise in terms of minimizing the MSE. Although the optimal solution relies on the unknown extended observability matrix to formulate $\Gamma$, it can be used with a preliminary model or a model set via minimax approaches.

On the other hand, similar to establishing the confidence region, the parameter $\Gamma$ used in the minimum MSE solution (\ref{eqn:minmse}) can be replaced by the data-driven estimate $\hat{\Gamma}_Z$ (\ref{eqn:est}) derived from the same signal matrix for an approximate solution. This leads to the minimum-MSE data-driven predictor, denoted as Algorithm~\ref{al:1}.

\begin{algorithm}[htb]
	\caption{The minimum-MSE data-driven predictor with stochastic data}
	\begin{algorithmic}[1]
	\State \textbf{Given: }signal matrix $Z$, noise model $\Sigma_Y, \Sigma_\text{yini}$, confidence level $p$.
	\State \textbf{Input: }$\mathbf{u}_{\text{ini}},\mathbf{y}_{\text{ini}},\mathbf{u}$.
	\State Calculate $\hat{\Gamma}_Z$ by (\ref{eqn:est}).
	\State Find $\mathbf{y}=\mathcal{F}_{Z}(\mathbf{u};\mathbf{u}_\text{ini},\mathbf{y}_\text{ini})$ by solving (\ref{eqn:minmse}) with $\Gamma=\hat{\Gamma}_Z$.
	\State Find $p$-confidence region $\mathcal{Y}$ by (\ref{eqn:bound}) with $\Gamma=\hat{\Gamma}_Z$.
	\State \textbf{Output: }$\mathbf{y}$, $\mathcal{Y}$.
	\end{algorithmic}
	\label{al:1}
\end{algorithm}

\section{Numerical Examples}
\label{sec:num}

In this section, numerical tests are conducted to illustrate the validity of the derived confidence region and the effectiveness of the proposed minimum-MSE algorithm. In the examples, stochastic data with i.i.d. noise are collected from one single experiment and used in $Z$ with a Page matrix construction. Unit Gaussian input sequences are used to generate the data.

First, we consider a simple two-dimensional example for illustration purposes. The prediction problem is to find the first two points ($L'=2$) in the step response of the following fourth-order system
\begin{equation}
    G_1(z) = \dfrac{0.1059(0.1z^4+z^3+0.5z^2)}{z^4-2.2z^3+2.42z^2-1.87z+0.7225}.
\end{equation}
The prediction conditions are $\mathbf{u}_\text{ini}=\mathbf{0}$, $\mathbf{y}_\text{ini}=\mathbf{0}$, and $\mathbf{u}=[1\ 1]^\mathsf{T}$. The following parameters are used: $L=10$, $L_0=8$, $M=80$, and noise level $\sigma^2=0.1$. A confidence level of $p=0.90$ is used in the following figures.

Figure~\ref{fig:1} compares the confidence regions obtained using model-based $\Gamma$ (\ref{eqn:gam}) (\textit{CR-MB}), data-driven $\hat{\Gamma}_\text{Sub}$ (\textit{CR-Sub}), $\hat{\Gamma}_\text{SMM}$ (\textit{CR-SMM}), and $\hat{\Gamma}_\text{WD}$ (\textit{CR-WD}). The confidence regions are tested on the minimum-MSE predictor with data-driven $\hat{\Gamma}_\text{SMM}$ (\textit{MSE-SMM}). 10 different realizations of the stochastic data are plotted. The results show that the data-driven formulations (\textit{CR-Sub}, \textit{CR-SMM}, and \textit{CR-WD}) obtain similar confidence regions, but are different from the model-based formulation. This is because the data-driven formulations with $\hat{\Gamma}_Z$ estimate the noise-free $\Gamma_Z$ that is different from the model-based $\Gamma$. Nevertheless, all the confidence regions are valid for this problem, since the true trajectory lies in the regions with high probability.

\begin{figure}[ht]
\centerline{\includegraphics[width=\columnwidth]{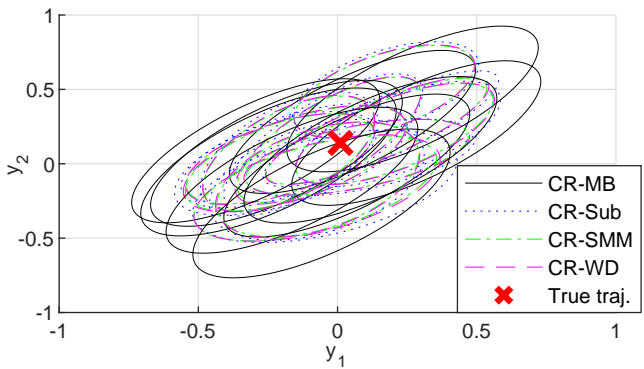}}
\caption{Comparison of different confidence region formulations ($p=0.90$) tested on the \textit{MSE-SMM} predictor with 10 different realizations of the stochastic data.}
\label{fig:1}
\end{figure}

Then, the sizes of the confidence regions are analyzed for different stochastic data-driven predictors. The following predictors are compared: 1) subspace predictor (\ref{eqn:sub}) (\textit{Sub}), 2) signal matrix model (\ref{eqn:smm}) (\textit{SMM}), 3) Wasserstein distance minimization (\ref{eqn:w}) (\textit{WD}), and 4) minimum-MSE predictor (Algorithm~\ref{al:1}) using model-based $\Gamma$ (\ref{eqn:gam}) (\textit{MSE-MB}), data-driven $\hat{\Gamma}_\text{Sub}$ (\textit{MSE-Sub}), $\hat{\Gamma}_\text{SMM}$ (\textit{MSE-SMM}), and $\hat{\Gamma}_\text{WD}$ (\textit{MSE-WD}). Figure~\ref{fig:2} shows the confidence regions of these stochastic predictors with model-based $\Gamma$ (\textit{CR-MB}). As can be seen from the figure, the existing algorithms (\textit{Sub}, \textit{SMM}, and \textit{WD}) have larger confidence regions compared to the minimum-MSE algorithms (\textit{MSE-MB} and \textit{MSE-SMM}). This illustrates the effectiveness of the proposed algorithm in improving prediction accuracy. In this example, the confidence regions of \textit{MSE-Sub} and \textit{MSE-WD} are very close to that of \textit{MSE-SMM}, so they are omitted in Figure~\ref{fig:2}.

\begin{figure}[ht]
\centerline{\includegraphics[width=\columnwidth]{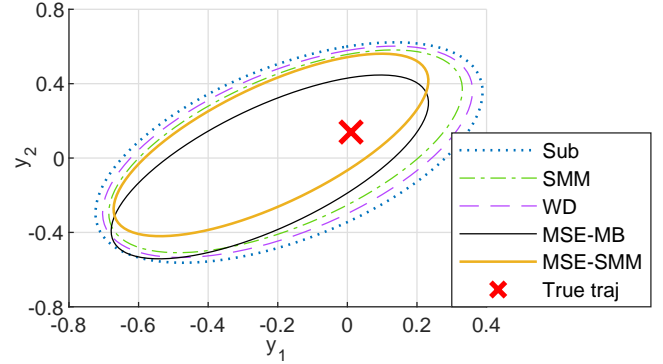}}
\caption{Comparison of different stochastic data-driven predictors in terms of the confidence regions ($p=0.90$) with model-based $\Gamma$ (\textit{CR-MB}).}
\label{fig:2}
\end{figure}

To quantitatively assess the derived confidence region and the minimum-MSE prediction algorithm, the following campaign of 1000 Monte Carlo simulations is set up. A bank of 1000 single-input, single-output systems are randomly generated by the \texttt{drss} command in \textsc{Matlab} with random numbers of states between 3 and 8. These random systems are normalized to have an $\mathcal{H}_2$-gain of 1. The prediction problem uses the following parameters: $L = 20$, $L_0 = 8$, $L' = 12$, and $M = 320$. The input $\mathbf{u}$ and the initial condition $(\mathbf{u}_\text{ini},\mathbf{y}_\text{ini})$ are selected randomly with a unit Gaussian distribution.

Table~\ref{tbl:1} compares the percentage of the simulations where the true response is in the confidence region, i.e., $\mathbf{y}_0^i\in\mathcal{Y}^i$ for the $i$-th simulation, for the model-based and different data-driven formulations. Two confidence levels $p=0.95$ and $p=0.99$ are selected. The noise level is selected as $\sigma^2=0.1$. The rows in Table~\ref{tbl:1} correspond to different predictors, whereas the columns correspond to different formulations of the confidence region. It can be seen from the table that the empirical confidence levels match the targeted $p$-value well with the model-based $\Gamma$ (\textit{CR-MB}) for all three predictors, where Theorem~\ref{thm:1} is satisfied exactly. With the data-driven estimates $\hat{\Gamma}_Z$, the confidence regions become marginally more conservative as the empirical confidence levels are slightly larger in Table~\ref{tbl:1}. The results of the three data-driven estimates (\textit{CR-Sub}, \textit{CR-SMM}, \textit{CR-WD}) are similar, which indicates that the confidence region is not very sensitive to the choice of $\hat{\Gamma}_Z$ estimation method.

\begin{table}[ht]
\centering
\caption{Empirical confidence levels of the confidence regions}
\renewcommand{\arraystretch}{1.1}
\begin{tabular}{ccccc}
\hline\hline
$p=0.95$         & CR-MB & CR-Sub & CR-SMM & CR-WD \\ \hline
Sub     & 97.1\% & 98.7\% & 98.4\% & 98.7\% \\
SMM     & 96.8\% & 97.4\% & 97.3\% & 97.3\% \\
MSE-SMM & 95.2\% & 96.4\% & 96.2\% & 96.4\% \\
\hline
$p=0.99$         & CR-MB & CR-Sub & CR-SMM & CR-WD \\ \hline
Sub     & 99.3\% & 100\% & 99.8\% & 99.9\% \\
SMM     & 99.2\% & 99.7\% & 99.7\% & 99.7\% \\
MSE-SMM & 99.0\% & 99.3\% & 99.2\% & 99.3\% \\
\hline\hline
\end{tabular}
\label{tbl:1}
\end{table}

Table~\ref{tbl:2} compares the empirical MSE of the predictors in the Monte Carlo simulations to the MSE estimated by (\ref{eqn:mse}) with the approximate data-driven confidence regions. The empirical MSE is computed as
\begin{equation}
    \text{MSE}_\text{emp}\left(\mathbf{y}-\mathbf{y}_0\right)=\frac{1}{N_s}\sum_{i=1}^{N_s} \norm{\mathbf{y}^i-\mathbf{y}_0^i}_2^2,
\end{equation}
where $\mathbf{y}^i$ and $\mathbf{y}_0^i$ are the predicted and the true responses of the $i$-th simulation respectively, and $N_s=1000$. Two different noise levels of $\sigma^2=0.1$ and $\sigma^2=1$ are considered. Similar to the observation from Table~\ref{tbl:1}, the estimated MSE is shown to be more conservative compared to the empirical ones for all three predictors. It is also observed that the region \textit{CR-SMM} is the less conservative among those tested here. However, the estimated MSE can correctly predict the relative error magnitudes of different predictors. This illustrates that the estimated MSE can be a good indicator of prediction accuracy, which motivates its use as the objective function in Algorithm~\ref{al:1}. Only three representative predictors are shown in Table~\ref{tbl:1} and Table~\ref{tbl:2} for clarity. The results of the other algorithms are similar.

\begin{table}[ht]
\centering
\caption{Comparison of the estimated and the empirical MSE}
\renewcommand{\arraystretch}{1.1}
\begin{tabular}{ccccc}
\hline\hline
$\sigma^2=0.1$   & Empirical & CR-Sub & CR-SMM & CR-WD \\ \hline
Sub     & 0.115 & 0.153 & 0.149 & 0.152 \\
SMM     & 0.099 & 0.142 & 0.137 & 0.140 \\
MSE-SMM & 0.096 & 0.136 & 0.131 & 0.134 \\
\hline
$\sigma^2=1$     & Empirical & CR-Sub & CR-SMM & CR-WD \\ \hline
Sub     & 1.106 & 1.529 & 1.485 & 1.511 \\
SMM     & 0.915 & 1.391 & 1.344 & 1.372 \\
MSE-SMM & 0.897 & 1.335 & 1.286 & 1.317 \\
\hline\hline
\end{tabular}
\label{tbl:2}
\end{table}

Finally, we compare the prediction accuracy of the predictors by the empirical MSE, under three different noise levels $\sigma^2=0.1$, $\sigma^2=0.5$, and $\sigma^2=1$. The results are shown in Table~\ref{tbl:3}. For all three noise levels, the minimum-MSE predictor with model-based $\Gamma$ (\textit{MSE-MB}) achieves the minimum empirical MSE. This is expected as \textit{MSE-MB} exactly optimizes for this objective as demonstrated in Proposition~\ref{prop:1}. However, the model-based $\Gamma$ is not available in practice. Among the other practical algorithms, Algorithm~\ref{al:1} with $\hat{\Gamma}_Z$ based on the signal matrix model (\textit{MSE-SMM}) has the smallest empirical MSE, with slightly better performance than the direct signal matrix model approach (\textit{SMM}). This result shows numerically that, with approximate data-driven formulations of $\hat{\Gamma}_Z$, the proposed minimum-MSE predictor still obtains a more accurate prediction than the existing algorithms.

\begin{table}[ht]
\centering
\caption{Comparison of the empirical MSE for different predictors}
\renewcommand{\arraystretch}{1.1}
\begin{tabular}{cccc}
\hline\hline
        & $\sigma^2=0.1$ & $\sigma^2=0.5$ & $\sigma^2=1$ \\ \hline
Sub     & 0.115 & 0.558 & 1.106 \\
SMM     & 0.099 & 0.476 & 0.915 \\
WD      & 0.113 & 0.548 & 1.091 \\
MSE-MB  & 0.094 & 0.435 & 0.833 \\
MSE-Sub & 0.097 & 0.464 & 0.908 \\
MSE-SMM & 0.096 & 0.460 & 0.897 \\
MSE-WD  & 0.097 & 0.462 & 0.902 \\
\hline\hline
\end{tabular}
\label{tbl:3}
\end{table}

\section{Conclusions}

In this paper, the prediction error of data-driven predictors with stochastic data is characterized statistically. The framework provides ellipsoidal confidence regions for various predictors. It also offers a novel optimal predictor that minimizes the mean-squared prediction error directly. In practice, both the confidence region and the minimum-MSE predictor can be implemented with data-driven approximations that show good accuracy numerically.

Both the derived confidence region and the minimum-MSE predictor can contribute to more reliable and effective applications of stochastic data-driven predictors to predictive control design with robustness guarantees on the satisfaction of safety-critical constraints.

\addtolength{\textheight}{-20cm}  

\bibliographystyle{IEEEtran}
\bibliography{IEEEabrv,refs}

\end{document}